\documentclass[a4paper]{article}

\usepackage[latin1]{inputenc}
\usepackage[english]{babel} 
\usepackage{amsmath,amsfonts,amssymb,amsthm,mathtools}
\usepackage{enumerate}
\usepackage{graphicx}
\usepackage{color}
\usepackage{relsize}

\usepackage{hyperref}

\newtheorem{thm}{Theorem}

\newtheorem{lemma}[thm]{Lemma}
\newtheorem{proposition}[thm]{Proposition}

\theoremstyle{definition}
\newtheorem{defn}[thm]{Definition}

\newcommand{\abs}[1]{\left\vert {#1} \right\vert}
\newcommand{\set}[1]{\left\{ {#1} \right\}}
\newcommand{\prt}[1]{\left( {#1} \right)}

\newcommand{\scal}[1]{\left< {#1} \right>}

\newcommand{\dv}[2]{\frac{\mathrm d{#1}}{\mathrm d{#2}}}
\newcommand{\pd}[2]{\frac{\partial{#1}}{\partial{#2}}}

\newcommand{\setC}{{\mathbb C}}

\newcommand{\sR}{{\mathbb R}}
\newcommand{\sC}{{\mathbb C}}

\newcommand{\lequi}{\ \ \Longleftrightarrow\ \ }

\newcommand{\A}{\mathcal{A}}

\renewcommand{\H}{\mathcal{H}}

\newcommand{\J}{\mathcal{J}}
\newcommand{\M}{\mathcal{M}}
\newcommand{\N}{\mathcal{N}}

\newcommand{\T}{\mathcal{T}}
\newcommand{\C}{\mathcal{C}}

\newcommand{\Sw}{\mathcal{S}}

\DeclareMathOperator{\Span}{span}

\DeclareMathOperator{\diag}{diag}

\DeclareMathOperator{\tr}{tr}
\DeclareMathOperator{\sgn}{sgn}

\newcommand{\DD}{\mathcal{D}}           
\newcommand{\ox}{\otimes}                  
\newcommand{\Dslash}{{\DD \mkern-11.5mu/\,}} 
\newcommand{\sH}{{\mathbb H}}
\newcommand{\x}{\times}
\newcommand{\vc}{\vcentcolon =}

\newcommand{\wA}{\widetilde{\mathcal{A}}}

\begin{document}

\title{\vspace{-0cm}{\bf Causality in noncommutative two-sheeted space-times}\vspace{0.5cm}}

\author{Nicolas Franco$^{a,b,}\footnote{ Current address: Department of mathematics, University of Namur, Rempart de la Vierge 8, 5000 Namur, Belgium}$ \and Micha{\l} Eckstein$^{c,b}$\vspace{0.3cm}}

\date{{\footnotesize $^a$ Copernicus Center for Interdisciplinary Studies{\footnote{ supported by a grant from the John Templeton Foundation}},\\
 ul. S{\l}awkowska 17, 31-016 Krak\'ow, Poland\\[0.2cm]
 $^b$ Faculty of Mathematics and Computer Science, Jagellonian University,\\
 ul. {\L}ojasiewicza 6, 30-348 Krak\'ow, Poland\\[0.2cm]
 $^c$  Faculty of Physics, Astronomy and Applied Computer Science, Jagiellonian University,\\
 ul. prof. Stanis{\l}awa {\L}ojasiewicza 11, 30-348 Krak\'ow, Poland\\[0.2cm]
 nicolas.franco@math.unamur.be \qquad michal.eckstein@uj.edu.pl}}

\maketitle

\begin{abstract}
We investigate the causal structure of two-sheeted space-times using the tools of Lorentzian spectral triples. We show that the noncommutative geometry of these spaces allows for causal relations between the two sheets. The computation is given in details when the sheet is a 2- or 4-dimensional globally hyperbolic spin manifold. The conclusions are then generalised to a point-dependent distance between the two sheets resulting from the fluctuations of the Dirac operator.
\end{abstract}

\section{Introduction}

Among the pseudo-Riemannian manifolds, the Lorentzian ones form a distinguished class because they can accommodate a causal structure. The latter has very deep consequences for physical models as it sets fundamental restrictions on the evolution of physical processes. On the mathematical side, the causal structure on a Lorentzian manifold $\M$ induces a partial order relation on the set of points of $\M$. The properties of this order have been studied by several authors (see for instance \cite{nachbin,Bes09,MinguzziCompactification,MinguzziUtilities}).

It turns out that the notion of a partial order can be generalised to the realm of noncommutative spaces \cite{Bes09}. This is to be understood as the existence of a partial order relation on the space of states of a, possibly noncommutative, $C^*$-algebra $A$. Via the Gelfand-Naimark theorem, it can be shown that a noncommutative partial order is equivalent to a usual partial order on $\mathrm{Spec}(A)$ whenever $A$ is commutative.

Inspired by these results, we proposed in \cite{CQG2013} an extended notion of a causal order suitable for noncommutative geometries (see also \cite{CC2014} for a less formal review). Our definition is embedded in the realm of Lorentzian spectral triples \cite{Stro} and recovers the classical causal structure for globally hyperbolic manifolds \cite[Theorem 7]{CQG2013}. We note that there exists an alternative approach based on the same ideas \cite{Bes09,Bes13,Bes14}, but focusing on more general orders without any specific relation to the metric (so not related to any Dirac operator).

To explore the properties of the proposed noncommutative causal structure we considered in \cite{SIGMA2014} a toy-model based on a noncommutive spectral triple $\big( \Sw(\sR^{1,1}) \ox M_2(\sC), L^2(\sR^{1,1},S) \ox \sC^2, \Dslash \ox 1 + \gamma \ox \diag\{d_1,d_2\} \big)$. It turned out that the triple at hand has a well-defined and highly non-trivial causal structure. It exhibits a number of interesting and unexpected features leading to constraints on the motion not only in the space-time component, but also in the internal space of the model. However, due to the complexity of the computations, we were not able to generalise our results to higher-dimensional, curved, space-times.

In this paper we investigate another toy-model --- a two-sheeted space-time --- based on a product of a globally hyperbolic space-time $\M$ and a finite spectral triple $\big( \A_F, \H_F, D_F \big)$, with $\A_F = \sC \oplus \sC$, $\H_F = \sC^2$ and $D_F = \left( \begin{smallmatrix} 0 & m \\ m^* & 0 \end{smallmatrix} \right)$. Since the algebra $\A_F$ is a commutative one and has only two pure states, the total space of physical states is isomorphic (at the set-theoretic level) to $\M \sqcup \M$. However, the resulting product geometry is non-trivial because the off-diagonal Dirac operator $D_F$ provides a link between the two sheets.

For this particular model we establish a procedure of determining the causal structure with $\M$ being a general even-dimensional globally hyperbolic manifold. We apply it explicitly in dimensions 2 and 4. Moreover, the adopted technique allows us to generalise the results to the case when the mass parameter $m$ is replaced with a complex scalar field.

The choice of the $\sC \oplus \sC$ model is also motivated on physical grounds. The noncommutative Standard Model of particle physics, based on the algebra $\A_{\mathrm{SM}} = \sC \oplus \sH \oplus M_3(\sC)$, is often described as a two-sheeted space-time \cite{SakellariadouDoubling}. Indeed, the space of pure states of the electroweak sector $\sC \oplus \sH$ consists of two points and although $P(M_3(\sC)) \cong \sC P^2$, all of its points are separated by an infinite distance as the Dirac operator $D_F$ commutes with the $M_3(\sC)$ part of the algebra \cite[Remark 5.1]{MartinettiMoyal}. Our results on the $M_2(\sC)$ model suggest that whenever two states are separated by an infinite distance, no causal relation between them is possible. Hence, the chosen finite algebra is a good toy-model for the full Standard Model based on $\A_{\mathrm{SM}}$. In this paper we will focus on the mathematical details of the model, postponing the discussion of the physical interpretation to a forthcoming one \cite{F10}.

The paper is organised as follows: In the next section we recollect the basic definitions and properties of Lorentzian spectral triples and noncommutative causal structures. In Section \ref{nc2s} we describe the features of the two-sheeted model and present the main result of the paper describing its causal structure. We work out in details the cases of space-time dimensions 2 and 4 in Sections \ref{sec2dim} and \ref{sec4dim} respectively. In Section \ref{secscalar} we study the impact of the inner fluctuations of the Dirac operator on the causal structure. We conclude with some general remarks on the applicability of the developed techniques to other almost commutative models.


\section{Causality in Lorentzian noncommutative geometry}\label{sec2}

As a prelude to the introduction of causality in noncommutative geometry, we need to recollect some elements of the theory of Lorentzian spectral triples. The usual definition of a spectral triple, as introduced by Connes \cite{C94,MC08}, allows only to deal with (typically compact) Riemannian spaces, while the notion of causality requires non-compact Lorentzian spaces. The first generalisation of the axioms to pseudo-Riemannian signatures was done in \cite{Stro} and led to various definitions \cite{Pas,F4,Verch11,Rennie12,F5}, which have however a common basis -- the Krein space. The theory of pseudo-Riemannian spectral triples is still very recent and undergoes an intense development.

Below we present a rather restrictive definition of a Lorentzian spectral triple following our previous works \cite{CC2014,F5}. It has the advantage of guaranteeing a signature of Lorentzian type and allows to recover a globally hyperbolic manifold in the commutative case. The following axioms can also be considered as a particular case of all other existing approaches.\\

\begin{defn}
\label{deflost}
A~\emph{Lorentzian spectral triple} is given by the data $(\mathcal{A},\widetilde{\mathcal{A}},\H,D,\mathcal{J})$
with:
\begin{itemize}\itemsep=10pt
\item A~Hilbert space $\H$.
\item A~non-unital dense $*$-subalgebra $\mathcal{A}$ of a $C^*$-algebra, with a~faithful representation as bounded operators on
$\H$.
\item A~preferred unitisation $\widetilde{\mathcal{A}}$ of $\mathcal{A}$, which is also a dense $*$-subalgebra of a $C^*$-algebra,
with a~faithful representation as bounded operators on $\H$ and such that $\mathcal{A}$ is an ideal of
$\widetilde{\mathcal{A}}$.
\item An unbounded operator $D$, densely defined on $\H$, such that:
\begin{itemize}\itemsep=3pt
\item $\forall a\in\widetilde{\mathcal{A}}$,   $[D,a]$ extends to a~bounded operator on $\H$, \item $\forall
a\in\mathcal{A}$,   $a(1 + \scal{D}^2)^{-\frac 12}$ is compact, with $\scal{D}^2 = \frac 12 (D D^* + D^*
D)$.
\end{itemize} 
\item A~bounded operator $\mathcal{J}$ on $\H$ with $\mathcal{J}^2=1$, $\mathcal{J}^*=\mathcal{J}$, $[\mathcal{J},a]=0$,
$\forall a\in\widetilde{\mathcal{A}}$ and such that:
\begin{itemize}\itemsep=3pt
\item $D^*=-\mathcal{J} D \mathcal{J}$ on $\text{Dom}(D) = \text{Dom}(D^*) \subset \H$;
\item there exist a densely defined self-adjoint operator $\mathcal{T}$ with $\text{Dom}(\mathcal{T}) \cap \text{Dom}(D)$ dense in $\H$ and with $\left(1+ \mathcal{T}^2 \right)^{-\frac{1}{2}}\in
\widetilde{\mathcal{A}}$, and a positive element $N\in\widetilde\A$ such that $\mathcal{J}  = -N [D,\mathcal{T}]$.\\
\end{itemize}
\end{itemize}
We say that a Lorentzian spectral triple is \emph{even} if there exists a~$\mathbb Z_2$-grading $\gamma$ of $\H$ such that
$\gamma^*=\gamma$, $\gamma^2=1$, $[\gamma,a] = 0$ $\forall a\in\widetilde{\mathcal{A}}$, $\gamma
\mathcal{J} =- \mathcal{J} \gamma$ and $\gamma D =- D \gamma $.
\end{defn}

The role of the operator $\mathcal{J}$, called fundamental symmetry, is to turn the Hilbert space $\H$ into a~Krein space on which the operator $iD$ is essentially self-adjoint~\cite{Stro,Bog}. As proved in \cite{CC2014,F5}, the condition $\mathcal{J}=-N[D,\mathcal{T}]$ guarantees the Lorentzian signature. More precisely, if a pseudo-Riemannian spectral triple is constructed from a pseudo-Riemannian manifold $\M$, then the condition $\mathcal{J}=-N[D,\mathcal{T}]$ implies that the signature of the metric is Lorentzian and moreover, the metric on $\M$ splits (i.e.~$\M$ is diffeomorphic to $\sR \times \Sigma$, where $\Sigma$ is a Cauchy surface). 

Let us now consider a~locally compact complete globally hyperbolic Lorentzian manifold $\mathcal{M}$ of dimension $n$ with a~spin structure $S$. By a~complete Lorentzian manifold we understand the following: there exists a~spacelike reflection --- i.e.~an automorphism $r$ of the tangent bundle respecting $r^2=1$, $g(r\cdot,r\cdot) = g(\cdot,\cdot)$~--- such that $\mathcal{M}$ equipped with a~Riemannian metric $g^r(\cdot,\cdot) \vc g(\cdot,r\cdot)$ is complete in the usual Lebesgue sense. Then one can always construct a~commutative Lorentzian spectral triple in the following way~\cite{CC2014,F5}:
\begin{itemize}
\item $\H_\M = L^2(\mathcal{M},S)$ is the Hilbert space of square integrable sections of a spinor bundle over~$\mathcal{M}$ (using the positive definite inner product on the spinor bundle).
\item $D_\M = -i(\hat c \circ \nabla^S) = -i e^{\;\;\mu}_a\gamma^{a} \nabla^S_\mu = -i \widetilde\gamma^{\mu} \nabla^S_\mu$ is the Dirac operator associated
with the spin connection~$\nabla^S$ (the Einstein summation convention is in use, $e^{\;\;\mu}_a$
stand for vielbeins, $\gamma^a$ -- for the flat gamma matrices and $\widetilde\gamma^\mu$ -- for the curved ones).\footnote{Conventions
used in the paper are $(-,+,+,+,\cdots)$ for the signature of the metric and
$\{\gamma^a,\gamma^b\}=2\eta^{ab}$ for the flat gamma matrices, with $\gamma^0$ anti-Hermitian and
$\gamma^a$ Hermitian for $a>0$. The curved gamma matrices respect the same Hermicity conditions and the relations $\{\widetilde\gamma^\mu,\widetilde\gamma^\nu\}=2 g^{\mu\nu}$. The notation $\widetilde\gamma^\mu$ is used in order to avoid any confusion.} 
\item \label{AM} $\mathcal{A}_\M \subset C^\infty_0(\mathcal{M})$ and $\widetilde{\mathcal{A}}_{\M} \subset
C^\infty_b(\mathcal{M})$ with pointwise multiplication are some appropriate sub-algebras\footnote{A typical choice for e.g.~a Minkowski space is $\A = \mathcal{S}({\mathbb R}^{1,n-1})$ the algebra of Schwartz functions (rapidly decreasing at infinity together with all derivatives) and $\widetilde\A = \mathcal{B}({\mathbb R}^{1,n-1})$ (bounded smooth functions with all derivatives bounded).} of the algebra of
smooth functions vanishing at infinity and the algebra of smooth bounded functions respectively.
{}$\widetilde{\mathcal{A}}_\M$ must be such that $\forall\, a\in\widetilde{\mathcal{A}}$, $[D,a]$ extends to
a~bounded operator on $\H$.
The representation is given by standard multiplication operators on $\H$: $(\pi(a) \psi)(x) = a(x) \psi(x)$
for all $x \in \mathcal{M}$.
\item $\mathcal{J}_\M=i\gamma^0$, where $\gamma^0$ is the first flat gamma matrix.\\
\end{itemize}

For a globally hyperbolic Lorentzian manifold, there always exists a global smooth time function $\T$ on $\M$ such that the metric splits \mbox{$g = - N^2 d^2\T + g_\T$}, where $g_\T$ is a collection of Riemannian metrics defined on the Cauchy hypersurfaces at constant $\T$ and $N$ is the lapse function. In such a case, the fundamental symmetry can be written as $\mathcal{J}=i\gamma^0=-N[D,\mathcal{T}]$ using the usual property of the Dirac operator $[D,\T]=-ic(d\T)$, so this construction respects the axioms of Definition \ref{deflost}. If $n$ is even, the $\mathbb Z_2$-grading is given by the chirality element: $\gamma_\M = \pm i^{\frac{n}{2} +1} \gamma^0 \cdots \gamma^{n-1}.$\\

Recall now that given a $C^*$-algebra $A$ one defines $S(A)$ -- the space of states, i.e.~positive linear functionals (automatically continuous) of norm one. For any $C^*$-algebra $S(A)$ is a convex set for the weak-$*$ topology and thus has a distinguished subset of extremal points $P(A)$ -- these are called the pure states on $A$. If $\A$ is a $*$-subalgebra of $A$, then one can also define $S(\A) \vc \{\phi\vert_{\A}, \; \phi \in S(A)\}$ and $P(\A)$ accordingly. Moreover, if $\A$ is dense in $A$, we have $S(\A) \simeq S(A)$ and $P(\A) \simeq P(A)$. When $A = C_0(X)$ for a locally compact Hausdorff topological space, pure states $P(A)$ are in one-to-one correspondence with the points of $X$ via the Gelfand-Naimark theorem.

 With the help of the data of a Lorentzian spectral triple one can equip the algebra $\wA$ with a causal structure as follows \cite[Definitions 4 and 5]{CQG2013}:

\begin{defn}\label{defncauscone}
Let $\mathcal{C}$ be a convex cone of all Hermitian elements $a \in \widetilde{\mathcal{A}}$ respecting
\begin{equation}\label{constcaus}
\forall \; \phi\in\H,\qquad\scal{\phi,\mathcal{J}[D,a]\phi}\leq0,
\end{equation}
where $\scal{\cdot,\cdot}$ is the inner product on $\H$.
If the following condition is fulfilled:
\begin{equation}\label{condcausal}
\overline{\Span_{{\mathbb C}}(\mathcal{C})}=\overline{\widetilde{\mathcal{A}}},
\end{equation}
where $\overline{\widetilde{\mathcal{A}}}$ denotes the $C^*$-completion of $\widetilde{\mathcal{A}}$, then $\mathcal{C}$ is called a \emph{causal cone}. It induces a partial order relation on
$S(\widetilde{\mathcal{A}})$, called \emph{causal relation}, by
\begin{align*}
\forall \omega,\eta\in S(\widetilde{\mathcal{A}}),\qquad\omega\preceq\eta\qquad\text{iff}\qquad\forall a\in\mathcal{C},\quad\omega(a)\leq\eta(a).
\end{align*}
\end{defn}

Definition \ref{defncauscone} is strongly motivated by the following result:

\begin{thm}[\cite{CQG2013}]\label{thmreconstruction}
Let $(\mathcal{A},\widetilde{\mathcal{A}},\H,D,\mathcal{J})$ be a commutative Lorentzian spectral
triple constructed from a complete globally hyperbolic Lorentzian manifold $\mathcal{M}$, and let us define the
following subset of pure states:
\begin{align*}
\mathcal{M}(\A) \vc \big\{\omega\in P(\widetilde{\mathcal{A}}) : \mathcal{A}\not\subset\ker\omega\big\}\cong P(\A) \cong \mathcal{M}.
\end{align*}
Then, the causal relation $\preceq$ on $S(\widetilde{\mathcal{A}})$ as defined in Definition \ref{defncauscone} restricted to $\mathcal{M}(\A)$ corresponds to the usual causal relation on
$\mathcal{M}$.
\end{thm}

The proof of this theorem can be found in \cite{CQG2013} and relies on the fact that the functions of the algebra respecting the condition \eqref{constcaus} are the \emph{causal functions} on the manifold, i.e.~real-valued functions which are non-decreasing along future directed causal curves. For this reason, we will use the name \emph{causal elements} for the elements of the algebra respecting \eqref{constcaus} also beyond the commutative case.

Definition \ref{defncauscone} does not depend on the choice of the fundamental symmetry $\J$ as it could equally well be formulated using the Krein space with its natural indefinite inner product \cite{PROC2015}. Also, the role of the unitisation is only technical as we are eventually interested in the causal relation on $\mathcal{M}(\A)$, which we regard as the space of physical states. On the other hand, in the commutative case the choice of the unitisation is equivalent to picking a suitable compactification of the space-time $\M$. The latter is directly related to an old-standing problem of attaching a `boundary' to a, possibly singular, space-time and extending the causal relation to it \cite{MinguzziCompactification}.

\section{Two-sheeted space-times}\label{nc2s}

In \cite{SIGMA2014} we started the exploration of causal structures in noncommutative geometry by studying almost commutative space-times. The latter are products of Lorentzian spectral triples based on globally hyperbolic manifolds and finite (Riemannian) spectral triples. These geometries provide a framework for models of particle physics \cite{WalterBook}.

In this paper we consider the simplest non-trivial finite spectral triple based on the algebra $\A_F = \sC \oplus \sC$. By the standard GNS construction we obtain a faithful representation $\pi$ of $\A_F$ on $\H_F = \sC^2$ by $\pi((a,b)) \psi \vc \left( \begin{smallmatrix} a & 0 \\ 0 & b \end{smallmatrix} \right) \psi$. From now on we shall omit the symbol of the representation and regard the elements of $\A_F$ as diagonal elements in $M_2(\sC)$ acting on $\H_F$ in a natural way. The most general Dirac operator $D_F$ in this setting is a Hermitian matrix in $M_2(\sC)$. This finite triple is even, with the grading $\gamma_F = \left( \begin{smallmatrix} 1 & 0 \\ 0 & -1 \end{smallmatrix} \right)$, but it does not admit a compatible reality structure unless $D_F = 0$ \cite[Proposition 3.1]{DungenED}.

\begin{defn}\label{def:twosheeted}
Let $\M$ be an even dimensional globally hyperbolic manifold. We call a \emph{two-sheeted space-time} the following Lorentzian spectral triple:
\begin{itemize}\itemsep=10pt
\item $\H = \H_\M \ox \H_F = L^2(\mathcal{M},S) \otimes {\mathbb C}^{2} \cong L^2(\mathcal{M},S) \oplus L^2(\mathcal{M},S)$,
\item $\mathcal{A} = \A_\M \otimes \A_F \cong \left(\begin{matrix}{\mathcal{A}}_\mathcal{M} &  0\\ 0 &  {\mathcal{A}}_\mathcal{M}\end{matrix}\right)$, 
\item $\widetilde{\mathcal{A}} = \wA_{\M} \otimes \A_F \cong \left(\begin{matrix}\wA_{\M} &  0\\ 0 &  \wA_{\M} \end{matrix}\right)$,
\item $D = D_\M \ox 1 + \gamma_\M \ox D_F = -i \widetilde\gamma^\mu\nabla^S_\mu \otimes 1 + \gamma_{\M} \ox D_F $, 
\item $\mathcal{J}= \J_\M \ox 1 = i\gamma^0 \otimes 1 = \left(\begin{smallmatrix}i\gamma^0 &  0\\ 0 &  i\gamma^0\end{smallmatrix}\right)$.
\end{itemize}
\end{defn}
This triple is even with $\gamma = \gamma_{\M} \ox \gamma_F$. We will say that the two-sheeted space-time has dimension $n$ if $\dim \M = n$. This is justified, as finite spectral triples are $0$-dimensional from the spectral point of view \cite{WalterBook}.

The results on causal structure do not depend on a particular choice of the algebras $\A_{\M}$ and $\wA_{\M}$, as long as they satisfy the constraints listed on page \pageref{AM} and the condition \eqref{condcausal}. For the purposes of this paper we make the following choices:
\begin{itemize}
\item $\A_\M = C^\infty_c(\M)$ -- the space of compactly supported smooth functions.
\item $\widetilde\A_\M = \Span_\setC({\C_\M})$, with $\C_\M \subset \mathcal{B}(\M)$ denoting the space of smooth causal functions with all derivatives bounded, where a bounded derivative of a function $a$ means here that $[D,a]$ extends to a~bounded operator on $\H$.
\end{itemize}

With these definitions, $(\mathcal{A}_\mathcal{M}, \widetilde{\mathcal{A}}_\mathcal{M},\H_\mathcal{M},D_\mathcal{M},\mathcal{J}_\mathcal{M})$  is an even Lorentzian spectral triple, which implies that the two-sheeted space-time $(\mathcal{A},\widetilde{\mathcal{A}},\H,D,\mathcal{J})$ is a Lorentzian spectral triple \cite[Theorem 1]{SIGMA2014}. We have chosen the algebra $\widetilde{\mathcal{A}}_\mathcal{M} = \Span_{{\mathbb C}}(\mathcal{C}_\mathcal{M})$ to meet the condition \eqref{condcausal}. The fact that $\widetilde{\mathcal{A}}_\mathcal{M}$ separates the points of $\M$ is a simple consequence of the Stone-Weierstrass theorem \cite{Bes09,CQG2013}:

\begin{proposition}
For every two-sheeted space-time, the condition $\overline{\Span_{{\mathbb C}}(\mathcal{C})}=\overline{\widetilde{\mathcal{A}}}$ is respected.
\end{proposition}

\begin{proof}
The proof relies on the complex version of the Stone--Weierstrass theorem, since $\widetilde\A=\Span_{{\mathbb C}}(\mathcal{C_\M}) \oplus \Span_{{\mathbb C}}(\mathcal{C_\M})$ corresponds to a sub-algebra of functions on some compactification of $\M \sqcup \M$ (Nachbin compactification \cite{nachbin,Bes09}). We only need to prove that the functions $\mathbf{a} \in \Span_{{\mathbb C}}(\mathcal{C_\M}) \oplus \Span_{{\mathbb C}}(\mathcal{C_\M})$ respecting $\forall \; \phi\in\H, \scal{\phi,\mathcal{J}[D,\mathbf{a}]\phi}\leq0$ separate the points of $\M \sqcup \M$. With $D_F = \left(\begin{smallmatrix}k & m\\m^* & l\end{smallmatrix}\right)$, $k,l \in \sR$, $m \in \sC$, the constraint \eqref{constcaus} can be rewritten and simplified:
\begin{equation}\label{eqdiracexp}
\forall \; \phi=(\phi_1,\phi_2)\in L^2(\mathcal{M},S) \oplus L^2(\mathcal{M},S),
\end{equation}
\[ \scal{\phi_1,\mathcal{J_\M}[D_\M,a]\phi_1} + \scal{\phi_2,\mathcal{J_\M}[D_\M,b]\phi_2} + (a-b) \scal{\phi,\mathcal{J_\M} \gamma_{\M} \otimes  \left(\begin{smallmatrix}0 & -m\\m^* & 0\end{smallmatrix}\right) \phi}  \ \leq0,
\]
with $\mathbf{a} = \left(\begin{smallmatrix}a & 0\\0 & b\end{smallmatrix}\right).$

On the strength of \cite[Theorem 11]{CQG2013}, each function $a \in \Span_{{\mathbb C}}(\mathcal{C_\M}) $ respecting $\scal{\phi_1,\mathcal{J_\M}[D_\M,a]\phi_1} \leq 0$ is actually in $\C_\M$, i.e.~$a$ is a causal functions (smooth, with all derivatives bounded) on $\M$, and the operator $\mathcal{J_\M}[D_\M,a]$ is (strictly) negative definite where the gradient of $a$ does not vanish. Hence, elements of the form $\mathbf{a} = \left(\begin{smallmatrix}a & 0\\0 & a\end{smallmatrix}\right)$ with $a \in \C_\M$ separate the points on the same copy of $\M$, and also different points on different copies. To show that two points on different copies of $\M$ with the same localisation (i.e.~two points $(p,0)$ and $(0,p)$ in $\M \sqcup \M$) can be separated, it is sufficient to choose locally $b = a + \epsilon$ with $a \in \C_\M$ having a non-vanishing gradient at $p$ and $\epsilon$ sufficiently small such that the inequality \eqref{eqdiracexp} remains valid. 
\end{proof}

\vspace{0.5cm}

Let us now turn to the space of states of the model at hand. As argued at the end of previous section we shall neglect the states in $P(\widetilde\A)$ coming from the compactification process. We define as in Theorem \ref{thmreconstruction} the space of physical pure states:
\begin{align*}
\mathcal{M}(\A) \vc \big\{\omega\in P(\widetilde{\mathcal{A}}) : \mathcal{A}\not\subset\ker\omega\big\}\cong P(\A).
\end{align*}

Since $\A_F$ has only two pure states and $P(\A) \cong \M \x \{\delta_1, \delta_2\} \cong \M \sqcup \M$, thus the name of a two-sheeted space-time. Although the algebra $\A$ (hence the topology) of the two-sheet space-time is commutative, its geometry is not if $D_F$ has a non-trivial diagonal part.

The full space of states is much larger and we will restrict to a special class of \emph{mixed} states defined as the convex combinations of two states with the same localisation:
\begin{align*}
\N(\A) \vc \M(\A_\M) \times S(\A_F) \cong \M \times [0,1] \subset S(\A).
\end{align*}
These states can be seen as covering the \emph{area} between the two sheets. A state in $\N(\A)$ will be denoted by $\omega_{p,\xi}$, with $p\in\M$, $\xi\in[0,1]$. Its evaluation on an element $\mathbf{a} = \left(\begin{smallmatrix}a & 0\\0 & b\end{smallmatrix}\right) \in \wA$ reads $\omega_{p,\xi}(\mathbf{a}) = \xi\, a(p) + (1-\xi)\, b(p)$. The extreme points, with $\xi=0$ or $\xi=1$, are precisely the pure states in $\M(\A)$.

We are now ready to investigate the causal structure of a two-sheeted space-time. We start with some general observations.

\begin{lemma}
\label{propcondcausal}
Let us take two states $\omega_{p,\xi},\omega_{q,\varphi}$ with $\varphi,\xi \in [0,1]$. Then, $\omega_{p,\xi} \preceq \omega_{q,\varphi}$ if and only if
$\forall\, \mathbf{a} = \left(\begin{smallmatrix} a & 0\\0 & b\end{smallmatrix}\right) \in \mathcal{C}$,
\[
\varphi\, a(q) - \xi\,a(p) + (1-\varphi)\, b(q) - (1-\xi)\, b(p) \geq 0.
\]
\end{lemma}
\begin{proof}
This is simply a rewriting of the condition $\forall a\in\mathcal{C}, \omega(a)\leq\eta(a)$ from Definition \ref{defncauscone}.
\end{proof}

\vspace{0.3cm}

We have then the first result:

\vspace{0.3cm}

\begin{proposition}
If the finite part Dirac operator $D_F$ is diagonal, then for any $p, q \in \M$, $\varphi,\xi \in [0,1]$ we have $\omega_{p,\xi} \preceq \omega_{q,\varphi}$ if and only if $\varphi = \xi$ and $p \preceq q$ in $\M$.
\end{proposition}

\begin{proof}
Using the same decomposition as in \eqref{eqdiracexp} with $D_F = \left(\begin{smallmatrix}k & 0\\0 & l\end{smallmatrix}\right)$, $k,l \in \sR$, we find:
\[ \forall \; \phi_1,\phi_2\in L^2(\mathcal{M},S),\quad \scal{\phi_1,\mathcal{J_\M}[D_\M,a]\phi_1} + \scal{\phi_2,\mathcal{J_\M}[D_\M,b]\phi_2} \  \ \leq0,
\]
which means that $\C \cong \C_\M \oplus \C_\M$.

Let us take two states $\omega_{p,\xi} \preceq \omega_{q,\varphi} $ with $\varphi,\xi \in [0,1]$ and consider the inequality 
\[
\varphi\, a(q) - \xi\,a(p) + (1-\varphi)\, b(q) - (1-\xi)\, b(p) \geq 0
\]
valid for all $a,b \in \C_\M$. Since a constant function is always a causal function, we can chose $b=0$ and $a=\pm 1$ to get $\varphi \geq \xi$ and $\varphi \leq \xi$, so there is no causal relation possible if $\varphi \neq \xi$.

Then under the hypothesis $\varphi = \xi$, the condition becomes $\varphi\, (a(q) - a(p)) + (1-\varphi)\, (b(q) - b(p)) \geq 0$, which is valid $\forall a,b\in\C_\M$ if and only if $p \preceq q$.
\end{proof}

\vspace{0.3cm}

The above result means that if $D_F$ is diagonal then no causal relation between the sheets is possible. Since the diagonal part of $D_F$ has no impact on the causal structure due to the fact that is commutes with the entire algebra $\A_F$, in the following we will restrict to $D_F$ of the form $D_F=\left(\begin{smallmatrix}0 & m\\m^* & 0\end{smallmatrix}\right)$, with $m \in \sC$.

\vspace{0.3cm}

\begin{proposition}\label{proppq}
Let us consider two states $\omega_{p,\xi}, \omega_{q,\varphi} \in \N(\A)$. If $\omega_{p,\xi} \preceq \omega_{q,\varphi}$, then $p \preceq q$.
\end{proposition}

\begin{proof}
By contradiction, let us suppose that $p \npreceq q$ so there exists a causal function $a \in \C_\M$ such that $a(q) < a(p)$. Using $\mathbf{a} = \left(\begin{smallmatrix}a & 0\\0 & a\end{smallmatrix}\right)$ which respects the inequality \eqref{eqdiracexp}, Lemma \ref{propcondcausal} gives $a(q) \geq a(p)$, so $\omega_{p,\xi} \npreceq \omega_{q,\varphi}$.
\end{proof}

\vspace{0.3cm}

This proposition implies no classical causality violation at the level of each sheet. In fact, the latter is general feature of causal structures of almost commutative space-times \cite{PROC2015}.

The remaining question is: can some causal relations be possible between two different sheets when the finite part Dirac operator is not diagonal? Surprisingly enough, the answer is positive. We now present the main result of the paper:

\begin{thm}\label{thm_MAIN}
Let $(\mathcal{A},\widetilde{\mathcal{A}},\H,D,\mathcal{J})$ be a two-sheeted space-time of dimension 2 or 4. Two states $\omega_{p,\xi}, \omega_{q,\varphi} \in \N(\A)$ are causally
related with $\omega_{p,\xi} \preceq \omega_{q,\varphi}$ if and only if  $p \preceq q$ on $\M$ and
\begin{align}\label{const2d}
l(\gamma)  \geq \frac{\abs{\arcsin\sqrt{\varphi}-\arcsin\sqrt{\xi}}}{\abs{m}},
\end{align}
where $l(\gamma)$ represents the length of a causal curve $\gamma$ going from $p$ to $q$ on $\M$.
\end{thm}

It is highly plausible that Theorem \ref{thm_MAIN} holds in an unaltered form in an arbitrary even number of dimensions, but the complexity of a rigorous proof grows significantly with the dimension (see also Section \ref{sec:conc}). In the next section we will provide a complete proof in the case of dimension 2, which is somewhat special as every 2-dimensional Lorentzian metric is (locally) conformally flat \cite[Example 7.9]{Nakahara}. Then, in Section \ref{sec4dim}, we will outline a general procedure of extending the proof to higher dimensions and provide the necessary details for dimension 4.

Before we pass on to the technical details, let us comment on the implications of Theorem \ref{thm_MAIN}.

Formula \eqref{const2d} clearly shows that for $m \neq 0$ there exist causal paths linking the two sheets. Indeed, if we take $\varphi=0$ and $\xi=1$, we find that two pure states localised at $p$ and $q$, but on different sheets are causally related if and only if $p$ and $q$ are causally related on $\M$ and 
\[
l(\gamma)  \geq \frac{\pi}{2\abs{m}},
\]
where $l(\gamma)$ represents the length of a causal curve $\gamma$ going from $p$ to $q$ on $\M$. This relation is graphically represented in Figure \ref{figcausalrel}. We note that the value $\frac{1}{\abs{m}}$ represents in fact the \emph{distance} between the two sheets as calculated using Connes' distance formula \cite{C94}. The condition on the length of the curve $\gamma$, which physically is the amount of proper time along $\gamma$, is then directly related to the distance between the two sheets. This result is similar to the one we obtained in the case of the finite algebra $\A_F = M_2(\setC)$ \cite{SIGMA2014}.

\begin{figure}[ht!] \centering
\includegraphics[width=10cm]{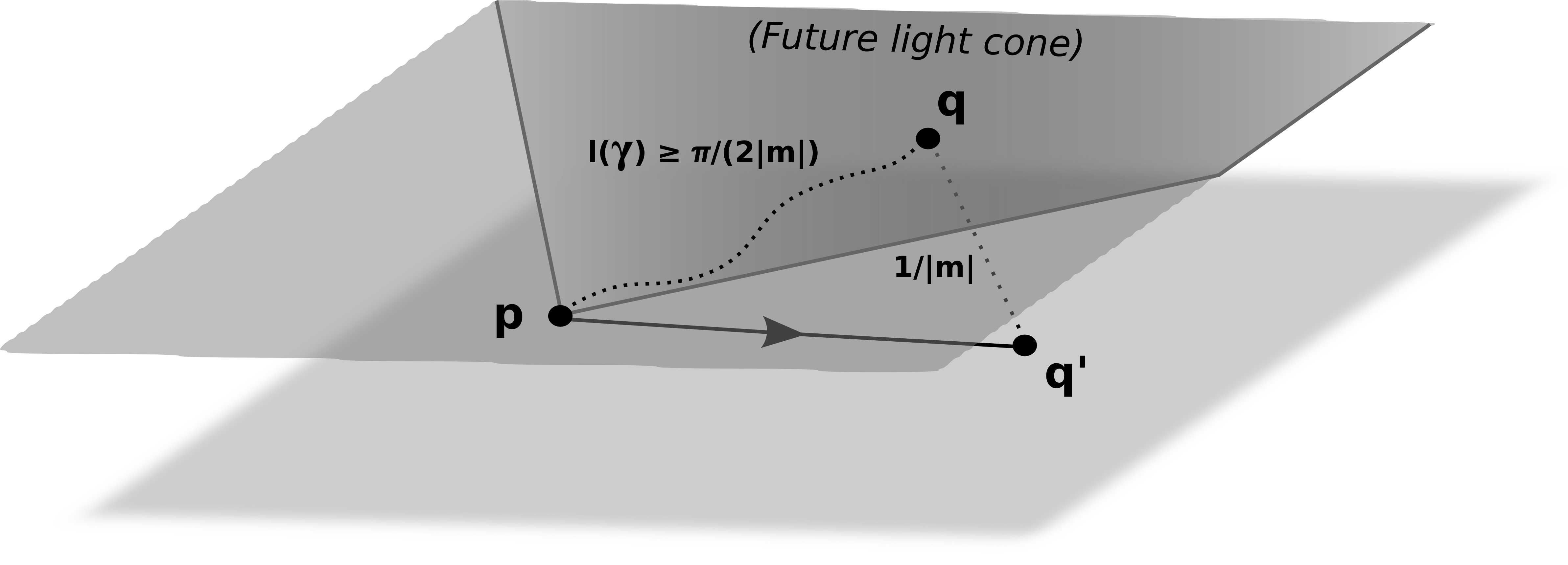}
\caption{Causal relation between the two sheets.}
\label{figcausalrel}
\end{figure}

As suggested by formula \eqref{const2d} a \emph{causal path} connecting the two sheets is not straight. For a two-dimensional Minkowski space-time it is in fact possible to draw a future cone (see \cite[Section 4]{CQG2013} for a precise definition) of a given state $\omega_{p,\xi}$, as represented in Figure \ref{figfutcon}.

\begin{figure}[ht!] \centering
\includegraphics[width=7cm]{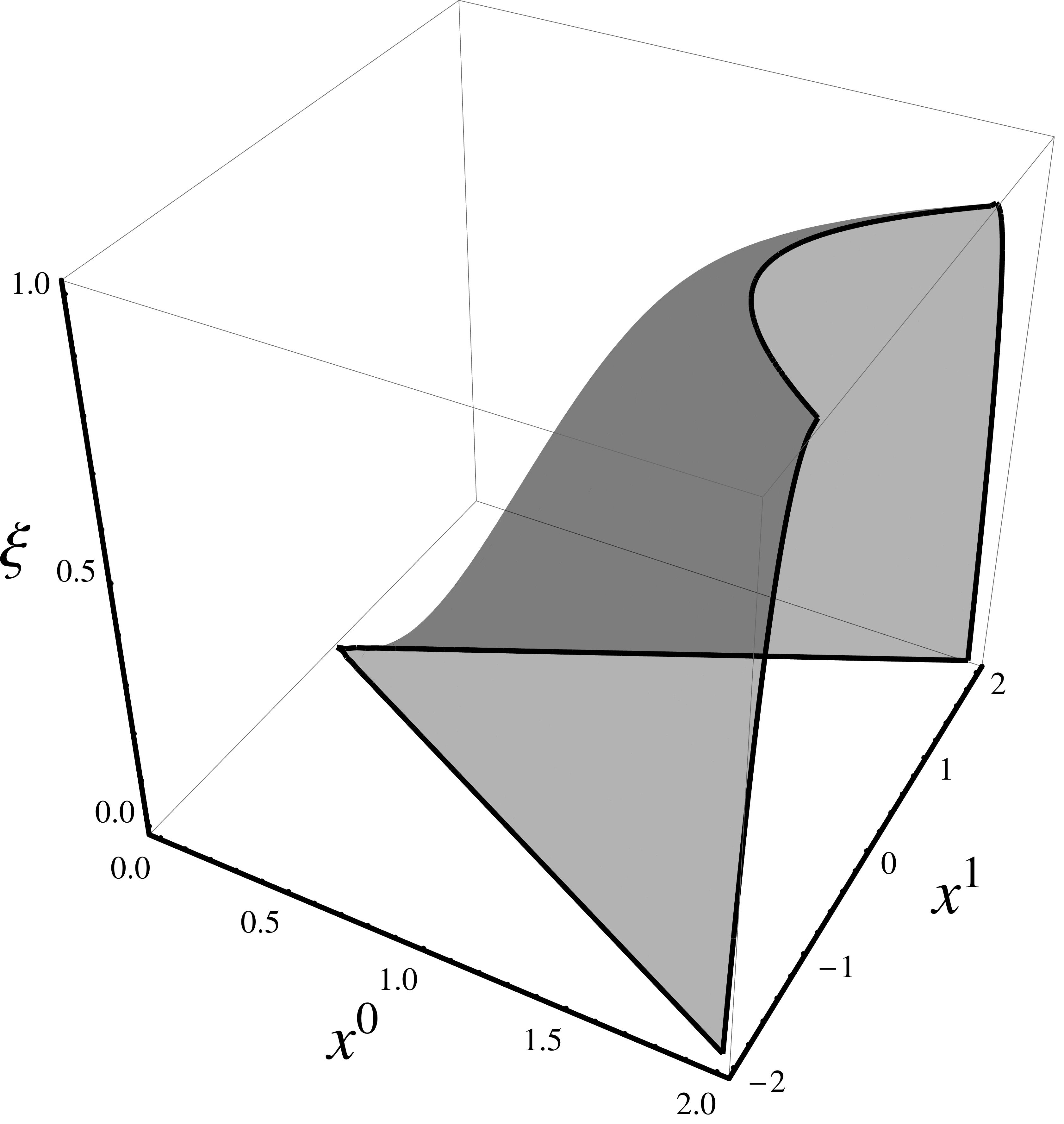}
\caption{The boundary of the future cone of the state $\omega_{(0,0),0}$ for $m = 1$ and $M = \sR^{1,1}$. All of the `causal paths' starting from $\omega_{(0,0),0}$ must lie under the plotted surface.}
\label{figfutcon}
\vspace*{0.8cm}
\end{figure}

\section{The two-dimensional case}\label{sec2dim}

In this section, we compute the causal structure for a two-sheeted space-time (recall Definition \ref{def:twosheeted}) of dimension 2. Let us consider a two-dimensional complete globally hyperbolic Lorentzian manifold $\M$ with a spin structure $S$. Since any two-dimensional metric is (locally) conformally flat \cite[Example 7.9]{Nakahara}, we have $g^{\mu\nu} = \Omega^2 \eta^{\mu\nu}$, with $\Omega$ a positive function on $\M$ and $\eta$ denoting the Minkowski metric. The Dirac operator in this setting reads
\begin{align*}
D = -i \widetilde\gamma^\mu\nabla^S_\mu \otimes 1 + \gamma^0 \gamma^1 \otimes  \left( \begin{smallmatrix} 0 & m \\ m^* & 0 \end{smallmatrix} \right),
\end{align*}
with $\widetilde\gamma^\mu = \Omega \gamma^{\mu}$ and some complex parameter $m$.

We will now prove the following:

\begin{thm}\label{thm2d}
Let $(\mathcal{A},\widetilde{\mathcal{A}},\H,D,\mathcal{J})$ be a two-sheeted space-time of dimension 2. Two states $\omega_{p,\xi}, \omega_{q,\varphi} \in \N(\A)$ are causally
related with $\omega_{p,\xi} \preceq \omega_{q,\varphi}$ if and only if  $p \preceq q$ on $\M$ and
\begin{align*}
l(\gamma)  \geq \frac{\abs{\arcsin\sqrt{\varphi}-\arcsin\sqrt{\xi}}}{\abs{m}},
\end{align*}
where $l(\gamma)$ represents the length of a causal curve $\gamma$ going from $p$ to $q$ on $\M$.
\end{thm}

The complete proof is based on Proposition \ref{prop2dimsc} for the sufficient condition, and on  Proposition  \ref{proppq} and Proposition \ref{prop2dimnc} for the necessary condition. 

We start with a simple technical Lemma:

\begin{lemma}
\label{propC}
Let $\mathbf{a} = \left(
\begin{smallmatrix}
a & 0
\\
0 & b
\end{smallmatrix}
\right) \in \widetilde{\mathcal{A}}$ be a~Hermitian element, then the following conditions are equivalent (we use here the notation $f_{,\mu} = \partial_\mu f = \pd{f}{x^\mu}$):
\begin{enumerate}[$(a)$] \itemsep=0pt \item \label{propC:a} $\mathbf{a}\in\mathcal{C}$, i.e.~$\forall\,\phi \in \H,\scal{\phi,\mathcal{J}[D,\mathbf{a}] \phi}
\leq 0$.
\item \label{propC:b} At every point of $\M$, the matrix
\[
 \left(
\begin{matrix} \Omega (a_{,0}+a_{,1}) &  0 &  0 &  -m(a-b)
\\
0 & \Omega (a_{,0}-a_{,1}) &  m(a-b) &  0
\\
0 &  m^*(a-b) &  \Omega (b_{,0}+b_{,1}) &  0
\\
-m^*(a-b) &  0 &  0 & \Omega (b_{,0}-b_{,1})
\end{matrix}
\right)
\]
is positive semidefinite.
\item \label{propC:c} At every point of $\M$, $\forall\, \phi_1,\phi_2,\phi_3,\phi_4 \in{\mathbb C},$
\begin{align*}
\Omega \cdot \Big( \abs{\phi_1}^2 (a_{,0}+a_{,1})+\abs{\phi_2}^2(a_{,0}-a_{,1}) + & \abs{\phi_3}^2( b_{,0}+b_{,1})+\abs{\phi_4}^2(b_{,0}-b_{,1}) \Big) \\
& \hspace*{0.5cm} \geq 2\Re\left\{(\phi_1^*\phi_4-\phi_2^*\phi_3)m\right\} \, (a-b).
\end{align*}
\end{enumerate}
\end{lemma}

\begin{proof}
Let us observe that the condition $\forall\,\phi \in \H$, $\scal{\phi,\mathcal{J}[D,\mathbf{a}] \phi} \leq 0$ is equivalent to having \mbox{$\mathcal{J}[D,\mathbf{a}] \leq 0$} at every point of $\M$, i.e.~$-\mathcal{J}[D,\mathbf{a}]$ is a~positive semi-definite matrix at every point of~$\M$. The second condition implies trivially the first one, and if the second one is false at some particular point $p\in\M$, then by continuity it is false on some open neighbourhood $U_p\subset\M$ and the first condition is false for some non-null spinor $\phi \in \H$ the support of which is included in~$U_p$.

From the Definition \ref{def:twosheeted} we read
\[
-\mathcal{J}[D,\mathbf{a}]=\left(
\begin{matrix}
-\gamma^0\widetilde\gamma^\mu\partial_\mu a &  i \,m(a-b) \gamma^0\gamma_\M
\\
-i \,m^*(a-b) \gamma^0\gamma_\M &  -\gamma^0\widetilde\gamma^\mu\partial_\mu b
\end{matrix}
\right) \, \cdot
\]

We chose the following \emph{chiral} representation of the 2-dimensional Dirac matrices:
\[
\gamma^0=  \left(
\begin{matrix}
0 &  i
\\
i &  0
\end{matrix}
\right),
\qquad
\gamma^1=  \left(
\begin{matrix}
0 &  -i
\\
i &  0
\end{matrix}
\right),
\qquad
\gamma_\M = \gamma^0\gamma^1=\left(
\begin{matrix}
-1 &  0
\\
0 &  1
\end{matrix}
\right)\cdot
\]

Since $\widetilde{\gamma}^{\mu} = \Omega \gamma^{\mu}$, the point \eqref{propC:b} follows.
The condition \eqref{propC:c} is just a~reformulation of the condition \eqref{propC:b} with an arbitrary vector $\phi =
(\phi_1,\phi_2,\phi_3,\phi_4) \in \sC^4$.

\end{proof}

\begin{lemma} \label{propcausalel}
If $\mathbf{a} = \left(
\begin{smallmatrix}
a & 0
\\
0 & b
\end{smallmatrix}
\right) 
\in\mathcal{C}$, then $a$ and $b$ are causal functions on $\M$.
\end{lemma}

\begin{proof}
Using Lemma \ref{propC}\eqref{propC:c} with	$\phi_1=1$ and $\phi_2=\phi_3=\phi_4=0$, we find that $a_{,0} + a_{,1} \geq 0$, and similarly $a_{,0} - a_{,1} \geq 0$ with $\phi_2=1$. So we are allowed to set:
\[
\abs{\phi_1}^2 = \tfrac{\Omega}{2} \left( a_{,0} - a_{,1} \right),\quad \abs{\phi_2}^2 = \tfrac{\Omega}{2} \left( a_{,0} + a_{,1} \right), \quad \phi_3=\phi_4=0,
\]
which gives  $\Omega^2 \left( -(a_{,0})^2 + (a_{,1})^2 \right) = g^{\mu\nu} a_{,\mu} a_{,\nu} \leq 0$ and the function $a$ has a timelike or null gradient. Setting $\phi_1=\phi_2=1$ and $\phi_3=\phi_4=0$ gives $\Omega a_{,0} \geq 0$, so the function is causal. A similar reasoning can be done for $b$.
\end{proof}

We can now prove the sufficient condition of Theorem \ref{thm2d}.

\begin{proposition}\label{prop2dimsc}
Let us suppose that $p \preceq q$ with $\gamma$ a~future directed timelike curve going from $p$ to $q$. We consider two states $\omega_{p,\xi}, \omega_{q,\varphi} \in \N(\A)$.  If 
\[
l(\gamma)  \geq \frac{\abs{\arcsin\sqrt{\varphi}-\arcsin\sqrt{\xi}}}{\abs{m}}
,\]
then $\omega_{p,\xi} \preceq \omega_{q,\varphi}$.
\end{proposition}

\vspace*{0.3cm}

\begin{proof}
Let us first take a~timelike curve $\gamma: [0,T] \to \M$, with $\gamma(0)=p$, $\gamma(T)=q$, $T>0$ and 
\[ l(\gamma) = \frac{\abs{\arcsin\sqrt{\varphi}-\arcsin\sqrt{\xi}}}{\abs{m}}.\]
We shall use the following notation:
\begin{align}
l^m_\gamma(t) & \vc  \int_0^t \abs{m} \sqrt{-g_{\gamma(s)}(\dot \gamma(s),\dot \gamma(s))}  \, ds \label{lm}\\
& = \int_0^t \abs{m} \big(\Omega(\gamma(s)) \big)^{-1} \sqrt{ (\dot\gamma^0(s))^2 - (\dot\gamma^1(s))^2}  \, ds \notag
\end{align}
representing the length of~$\gamma$ restricted to the interval $[0,t]$ and multiplied by $\abs{m}$, and denote $\sigma \vc \sgn\left(\arcsin\sqrt{\varphi}-\arcsin\sqrt{\xi}\right) \in \set{+1,-1}$.

Our hypothesis is then 
\[\sigma\,l^m_\gamma(T) = \arcsin\sqrt{\varphi}-\arcsin\sqrt{\xi}.\]

We define the following function for $t\in[0,T]$:
\begin{align}\label{chi}
\chi(t) \vc \sin^2\left( \sigma\,l^m_\gamma(t) + \arcsin \sqrt{\xi}  \right),
\end{align}
which respects $\chi(0) = \xi$ and $\chi(T) = \varphi$.\\

Let us take an arbitrary $\mathbf{a} = \left(\begin{smallmatrix} a & 0\\0 & b\end{smallmatrix}\right) \in \mathcal{C}$. We have, using the second fundamental theorem of calculus:
\begin{eqnarray}
&&\varphi\,a(q) - \xi\, a(p) + (1-\varphi)\, b(q) - (1-\xi)\, b(p) \label{eqgengoal}\\
&=&\chi(T)\, a(\gamma(T)) - \chi(0)\, a(\gamma(0)) + (1-\chi(T)) \,b(\gamma(T)) - (1-\chi(0)) \,b(\gamma(0))\nonumber\\
&=& \int_0^T  \left(\chi\cdot (a \circ \gamma) \right)^\prime(t) + \left((1-\chi)\cdot (b \circ \gamma) \right)^\prime(t) \;dt\nonumber\\
&=& \int_0^T  \chi^\prime (t)\, (a \circ \gamma) (t) + \chi(t)\,(a\circ\gamma)^\prime(t) - \chi^\prime (t)\, (b \circ \gamma) (t) + (1-\chi(t))\,(b\circ\gamma)^\prime(t)\;dt\nonumber\\
&=& \int_0^T  \chi^\prime (t)\, (a-b)(\gamma(t)) \;dt + \int_0^T \left[ \chi(t)\,(a\circ\gamma)^\prime(t)  + (1-\chi(t))\,(b\circ\gamma)^\prime(t) \right] \;dt. \label{eqgoal}
\end{eqnarray}

According to Lemma \ref{propcondcausal}, we need to prove that \eqref{eqgengoal} is always non-negative.\\

At first, we consider the integrand of the second term in \eqref{eqgoal}:
\begin{align}
\chi \cdot (a\circ\gamma)^\prime  + (1-\chi) \cdot (b\circ\gamma)^\prime & = \chi \cdot ( \dot \gamma^0 a_{,0} + \dot \gamma^1 a_{,1}) + (1-\chi) \cdot ( \dot \gamma^0 b_{,0} + \dot \gamma^1 b_{,1}) \nonumber\\
&= \Omega \cdot \Big( \chi \cdot \left[ \lambda_1 (a_{,0}+a_{,1}) + \lambda_2 (a_{,0}-a_{,1}) \right] \nonumber\\
& \qquad + (1-\chi) \cdot \left[ \lambda_1 (b_{,0}+b_{,1}) + \lambda_2 (b_{,0}-b_{,1}) \right] \Big)\, , \label{eq2ft}
\end{align}
where the arguments $t$ and $\gamma(t)$ are omitted in order to have more readable expressions and using the following functions
\begin{equation}\label{defl1l2}
\lambda_1 = \frac{\dot{\gamma}^{0} + \dot{\gamma}^{1}}{2 \Omega} \qquad \text{and}\qquad \lambda_2 = \frac{\dot{\gamma}^{0} - \dot{\gamma}^{1}}{2 \Omega},
\end{equation}
which are positive since $\gamma$ is timelike. 

Then, by applying Lemma \ref{propC}$(c)$ in \eqref{eq2ft} with
\[ \phi_1 = \sqrt{\chi \lambda_1},\qquad \phi_2 = \sqrt{\chi \lambda_2},\qquad\phi_3 = -\sqrt{(1-\chi)  \lambda_1} \,e^{i \delta},\qquad\phi_4 = \sqrt{(1-\chi) \lambda_2} \,e^{i\delta},\]
we find the inequality:
\begin{align}
\chi \cdot(a\circ\gamma)^\prime  + (1-\chi) \cdot (b\circ\gamma)^\prime & \geq \sqrt{\chi(1-\chi)}  \sqrt{\lambda_1 \lambda_2}\; 4 \Re\left\{e^{i\delta}m(a-b)\right\}  \label{ineq2d}\\
&= 2\sqrt{\chi(1-\chi)} \abs{m}  2\sqrt{\lambda_1 \lambda_2}\;  \abs{a-b} \label{eqterm2}
\end{align}
if $\delta$ is chosen such that $e^{i\delta}m(a-b) =  \abs{m}\abs{a-b}$.

Then, by looking at the integrand of the first term in \eqref{eqgoal} we find:
\begin{align}
\chi^\prime \cdot (a-b) & = \sigma \;2 \sin\left( \sigma\,l^m_\gamma + \arcsin \sqrt{\xi}  \right) \cos\left( \sigma\,l^m_\gamma + \arcsin \sqrt{\xi}  \right)  \,\left(l^m_\gamma\right)^\prime (a-b)\nonumber\\
&= \sigma \;2 \sqrt{\chi(1-\chi)} \,\left(l^m_\gamma\right)^\prime (a-b)  \label{eqterm1}.
\end{align}

In order to show that \eqref{eqgoal} is non-negative using \eqref{eqterm2} and \eqref{eqterm1}, it is sufficient to check that:
\begin{equation} \label{rell1l2}
\abs{m} 2\sqrt{\lambda_1 \lambda_2} = \abs{m} \Omega^{-1} \sqrt{(\dot\gamma^0)^2 - (\dot\gamma^1)^2} = \left(l^m_\gamma\right)^\prime . 
\end{equation} 

The proof will be complete if we show that the result remains true under the general hypothesis
\[ l(\gamma) \geq \frac{\abs{\arcsin\sqrt{\varphi}-\arcsin\sqrt{\xi}}}{\abs{m}}.\]
From the transitivity of the causal order, it is sufficient to prove that $\omega_{p,\varphi} \preceq \omega_{q,\varphi}$ (i.e.~with $\varphi=\xi$) for every future directed timelike curve with $l(\gamma) \geq 0$. In this case, \eqref{eqgengoal} becomes:
\[
\varphi\,(a(q) - a(p)) + (1-\varphi)\, (b(q) -  b(p)) \geq 0,
\]
which is non-negative since $a$ and $b$ are causal functions from Lemma \ref{propcausalel}.
\end{proof}

The necessary condition of Theorem \ref{thm2d} is given by this following proposition:

\begin{proposition}\label{prop2dimnc}
Let us suppose that $p \preceq q$ and consider two states $\omega_{p,\xi}, \omega_{q,\varphi} \in \N(\A)$ with $\varphi\neq\xi$. If $\omega_{p,\xi} \preceq \omega_{q,\varphi}$, then there exists a~future directed timelike curve $\gamma$
going from $p$ to $q$ such that 
\[
l(\gamma)  \geq \frac{\abs{\arcsin\sqrt{\varphi}-\arcsin\sqrt{\xi}}}{\abs{m}}\, .
,\]
\end{proposition}

\begin{proof}
Let us consider an arbitrary future directed timelike curve $\gamma$ such that $\gamma(0)=p$ and
$\gamma(T)=q$ and suppose by contradiction that 
\begin{equation}\label{assum}
l^m_\gamma(T) <\sigma\, \left( \arcsin\sqrt{\varphi}-\arcsin\sqrt{\xi} \right),
\end{equation}
with $l^m_\gamma(t) =  \int_0^t \abs{m} \sqrt{-g_{\gamma(s)}(\dot \gamma(s),\dot \gamma(s))}  ds$ and $\sigma = \sgn\left(\arcsin\sqrt{\varphi}-\arcsin\sqrt{\xi}\right) $.

Using Lemma \ref{propcondcausal}, we need to demonstrate that there exists at least one causal element $\mathbf{a} = \left(
\begin{smallmatrix}
a & 0
\\
0 & b
\end{smallmatrix}
\right) \in \mathcal{C}$ such that 
\begin{equation}\label{goalstrict}
\varphi\,a(q) - \xi\, a(p) + (1-\varphi)\, b(q) - (1-\xi)\, b(p) < 0.
\end{equation}

Such an element is explicitly defined along the curve $\gamma$ (and we consider its implicit smooth extension to $\M$) in the following way:
\begin{eqnarray}
a(\gamma(t)) &=&  - \frac 12 \cot( l^m_\gamma(t) + \sigma \arcsin\sqrt{\xi} + \epsilon), \nonumber \\
b(\gamma(t)) &=&   \frac 12 \tan( l^m_\gamma(t) + \sigma \arcsin\sqrt{\xi} + \epsilon), \label{defab}
\end{eqnarray}  
where $\epsilon\geq0$ is chosen, if needed, such that the argument lives inside the interval $]0,\frac \pi 2[$ or $]-\frac \pi 2,0[$ (such an $\epsilon$ always exists since $\forall t\in[0,T], \abs{l^m_\gamma(t) + \sigma \arcsin\sqrt{\xi}} < \frac \pi 2$ by the assumption).

The proof that $\mathbf{a} = \left(
\begin{smallmatrix}
a & 0
\\
0 & b
\end{smallmatrix}\right)$ respects the conditions of a causal element is a technical computation detailed in Appendix A.

Let us first consider the non-pure states case with $\varphi,\xi \in ]0,1[$, where we can choose $\epsilon=0$. We have, using our assumption \eqref{assum} and the increasing behaviour of the tangent and $-$cotangent functions:
\begin{eqnarray*}
\varphi\, a(q) &<& -\varphi\, \frac 12  \cot( \sigma\arcsin\sqrt{\varphi}) = -\sigma \frac 12 \sqrt{\varphi(1-\varphi)},\\
-\xi\, a(p) &=& \xi\, \frac 12  \cot( \sigma\arcsin\sqrt{\xi}) = \sigma \frac 12 \sqrt{\xi(1-\xi)},\\
(1-\varphi)\, b(q) &<& (1-\varphi)\, \frac 12  \tan( \sigma\arcsin\sqrt{\varphi}) = \sigma \frac 12 \sqrt{\varphi(1-\varphi)},\\
-(1-\xi)\, b(p) &=& -(1-\xi)\, \frac 12  \tan( \sigma\arcsin\sqrt{\xi}) = -\sigma \frac 12 \sqrt{\xi(1-\xi)},
\end{eqnarray*} 
and therefore the strict inequality \eqref{goalstrict} is respected.

The cases involving pure states must be treated separately, using $\epsilon>0$ such that $l^m_\gamma(t) + \sigma \arcsin\sqrt{\xi} + \epsilon < \sigma \arcsin\sqrt{\varphi}$:
\begin{itemize}
\item for $\xi=0$ and $\varphi\neq1$,
\begin{flalign*}
&\eqref{goalstrict} < - \frac 12 \sqrt{\varphi(1-\varphi)} - 0 +  \frac 12 \sqrt{\varphi(1-\varphi)} - \frac 12 \tan \epsilon < 0;&
\end{flalign*}	
\item for $\xi\neq0$ and $\varphi=1$,
\begin{flalign*}
&\eqref{goalstrict} = - \frac 12  \cot( l^m_\gamma(t) + \arcsin\sqrt{\xi} + \epsilon) +  \frac 12 \sqrt{\xi(1-\xi)} + 0 - \frac 12 \sqrt{\xi(1-\xi)}< 0;&
\end{flalign*}
\item for $\xi=0$ and $\varphi=1$,
\begin{flalign*}
&\eqref{goalstrict} = - \frac 12  \cot( l^m_\gamma(t) + \arcsin\sqrt{\xi} + \epsilon) -0 + 0 - \frac 12 \tan \epsilon< 0&
\end{flalign*}
\end{itemize}
and the symmetric cases (with $\sigma=-1$) are treated in a similar way.
\end{proof}

With this proposition, Theorem \ref{thm2d} is proven for curves which are everywhere timelike. When the curve is null or partially null, the result follows from continuity and transitivity of the causal order. In fact, if a point $q$ is localised at the boundary of the light cone of a point $p$ (i.e.~every causal curve $\gamma$ from $p$ to $q$ has length $l(\gamma)=0$), the only possibility to get a causal relation between states $\omega_{p,\xi} \preceq \omega_{q,\varphi}$ implies $\varphi=\xi$ (so no movement in the internal space, or no causal relation between the two sheets in the case of pure states). Indeed, if a causal relation was possible with $\varphi\neq\xi$, then there would exist a point $q'$ in the future of $q$ and $p$ and two (maximal) timelike curves $\gamma_p$ and $\gamma_q$ from $p$ and $q$ respectively to $q'$, with similar lengths, i.e.~$\abs{l(\gamma_p)-l(\gamma_q)} < \delta$, for some $\delta > 0$, due to the continuity of the Lorentzian distance in globally hyperbolic space-times \cite{Beem}. Since $q'$ can be infinitely near $q$, $\delta$ can be small enough to yield $\abs{m} \delta < \abs{\arcsin\sqrt{\varphi}-\arcsin\sqrt{\xi}}$. Then, the movement in the internal space would be more important along the curve $\gamma$ followed by $\gamma_q$ than along $\gamma_p$, which contradicts the transitivity of the causal order between states as proved in \cite{CQG2013}.

\section{The four-dimensional case}\label{sec4dim}

We now outline how to extend the proof of Theorem \ref{thm2d} to higher dimensional two-sheeted space-times. The key element was point \eqref{propC:c} of Lemma \ref{propC}, which translates the condition \eqref{condcausal} into an inequality linear in the derivatives of $a$ and $b$. The latter leads to inequalities \eqref{ineq2d}, \eqref{eqterm2}, which then can be integrated along a causal curve $\gamma$ (compare Formulae \eqref{eqgengoal} and \eqref{eqgoal}).

Let us now see how a suitable `linearised' version of \eqref{condcausal} can be obtained in higher dimensions.

For $\mu=0,\dots,n-1$, we define the following operators on $\H_\M$:
\[  
V^{\mu} \vc - \gamma^0 \widetilde \gamma^\mu,
\]
where $\gamma^0$ is the first flat gamma matrix of the spin structure and $\widetilde \gamma^\mu$ are the curved gamma matrices. We can notice that those operators are Hermitian (positive for $V^0$) and respect $(V^0)^2=-g^{00}$, $V^0V^\mu = V^\mu V^0$ and $V^\mu V^\nu + V^\nu V^\mu = 2 g^{\mu\nu}$ for $\mu, \nu > 0$. We also define 
\[
V \vc - \gamma^0  \gamma,
\]
where $\gamma$ is the chirality element, respecting $V^2 = 1$, commuting with $V^0$ and anti-commuting with $V^\mu$ for $\mu>0$.

Using such notation, the matrix involved in the constraint \eqref{constcaus} reads
\begin{equation}\label{matrixEN}
-\J [D,\mathbf{a} ] =  \left( \begin{matrix}
V^{\mu} a_{,\mu} & -iV m(a-b) \\
 iV m^*(a-b) & V^{\mu} b_{,\mu}  
\end{matrix} \right),
\end{equation}
with $\mathbf{a} = \left(
\begin{smallmatrix}
a & 0
\\
0 & b
\end{smallmatrix}
\right) \in \mathcal{A}$ Hermitian and using Einstein summation convention.

For an arbitrary future directed timelike curve $\gamma$ on $\M$ and a particular point $p=\gamma(s)$ on this curve, we define the quantities:
\[
v^{\mu} \vc \dot\gamma(s)^\mu,
\]
as the components of the vector tangent to the curve at $p$.

Let us suppose that one can find all the complex vectors $\psi \in \setC^{2^{ \lfloor\frac{n}{2}\rfloor}}$ respecting the following system of equations:
\begin{equation}\label{eqquantum}
\psi^* V^\mu (p) \psi = v^\mu,\quad \forall \mu=1,\dots,n-1.
\end{equation}
We remark that the system \eqref{eqquantum} can be interpreted from a quantum mechanical point of view by regarding that the quantities $v^\mu$ as the expectation values of the operators $V^\mu$ in the \emph{states} $\psi$. We also note that, the bigger $n$ is, the more underdetermined the system \eqref{eqquantum} is.

The system \eqref{eqquantum} can be rewritten using the vielbeins formalism and a flat version of the operators $V^\mu$:
\begin{equation}\label{eqquantumflat}
\psi^* V^\mu (p) \psi = \psi^* e^{\;\;\mu}_a V^a (p) \psi  = v^\mu \quad \lequi \quad \psi^* V^a (p) \psi = w^a, 
\end{equation}
where $V^a= - \gamma^0 \gamma^a$ and $w^a = e_{\;\;\mu}^a v^\mu$.

Now, let $\psi_1$ and $\psi_2$ be two solutions of system \eqref{eqquantum} and let $\chi$ be as in \eqref{chi}. Having fixed a point on $\M$, we compute the expectation value of the matrix $-\J[D,\mathbf{a}]$ in the state defined by the vector $\phi = (\sqrt{\chi} \psi_1, \sqrt{1-\chi} \psi_2)$, which should be non-negative for any $\mathbf{a} \in \mathcal{C}$. We obtain
\begin{align*}
-\phi^* \J[D,\mathbf{a}] \phi & = \chi \psi_1^* V^{\mu} \psi_1 a_{,\mu} - i \sqrt{\chi(1-\chi)} (a-b) m \psi_1^* V \psi_2 + \\
& \quad + (1-\chi) \psi_2^* V^{\mu} \psi_2 b_{,\mu} + i \sqrt{\chi(1-\chi)} (a-b) m^* \psi_2^* V \psi_1 \\
& = \chi v^{\mu} a_{,\mu} + (1-\chi) v^{\mu} b_{,\mu} - 2 \sqrt{\chi(1-\chi)} (a-b) \Re \{ i m \psi_1^* V \psi_2 \},
\end{align*}
which yields $\mathbf{a} = \left(
\begin{smallmatrix}
a & 0
\\
0 & b
\end{smallmatrix}
\right) \in \mathcal{C} \quad \Longleftrightarrow $
\begin{align}\label{ineq}
 \forall\,{p \in \M} \quad \chi v^{\mu} a_{,\mu} + (1-\chi) v^{\mu} b_{,\mu} \geq 2 \sqrt{\chi(1-\chi)} (a-b) \Re \{ i m \psi_1^* V \psi_2 \}.
\end{align}

Finally, one should use the degrees of freedom of $\psi_1$ and $\psi_2$ to maximise the RHS of the above inequality, as in Formula \eqref{ineq2d}.

Having outlined the general procedure we have to stress that both solving the system \eqref{eqquantum} and finding a configuration that maximises the inequality in \eqref{ineq} can be rather cumbersome for large dimensions. However, in the case of dimension 4 to which we shall now turn, we were able to carry on the computations.

The Dirac operator in dimension 4 reads
\begin{align*}
D = -i \widetilde\gamma^\mu\nabla^S_\mu \otimes 1 + i \gamma^0 \gamma^1 \gamma^2 \gamma^3 \otimes  \left( \begin{smallmatrix} 0 & m \\ m^* & 0 \end{smallmatrix} \right),
\end{align*}
with $\widetilde\gamma^\mu = e^{\;\;\mu}_a\gamma^{a}$ and some complex parameter $m$.

\begin{thm}\label{thm4d}
Let $(\mathcal{A},\widetilde{\mathcal{A}},\H,D,\mathcal{J})$ be a two-sheeted space-time of dimension 4. Two states $\omega_{p,\xi}, \omega_{q,\varphi} \in \N(\A)$ are causally
related with $\omega_{p,\xi} \preceq \omega_{q,\varphi}$ if and only if  $p \preceq q$ on $\M$ and
\[
l(\gamma)  \geq \frac{\abs{\arcsin\sqrt{\varphi}-\arcsin\sqrt{\xi}}}{\abs{m}},\]
where $l(\gamma)$ represents the length of a causal curve $\gamma$ going from $p$ to $q$ on $\M$.
\end{thm}

\vspace*{0.1cm}

\begin{proof}
As in Section \ref{sec2dim}, the proof will be done for timelike relations, and extended to causal relations by continuity. Proposition \ref{proppq} can also be used to guarantee the usual causal relation $p \preceq q$.

The explicit formulas for the operators $V^a$ present in equations \eqref{eqquantumflat} are given in the Appendix B. Using the notation $\psi = (r_i e^{i \theta_i})_{i=1,2,3,4} \in \setC^4$, the system \eqref{eqquantumflat} becomes:
\begin{eqnarray}
r_1^2 + r_2^2 + r_3^2 + r_4^2 &=& w^0 \label{sysv0},\\
2 r_1 r_2 \cos(\theta_2- \theta_1) - 2 r_3 r_4 \cos(\theta_4- \theta_3) &=& -w^1 \label{sysv1},\\
2 r_1 r_2 \sin(\theta_2- \theta_1) - 2 r_3 r_4 \sin(\theta_4- \theta_3) &=& -w^2 \label{sysv2},\\
r_1^2 - r_2^2 - r_3^2 + r_4^2 &=& -w^3. \label{sysv3}
\end{eqnarray}

The equations \eqref{sysv0} and \eqref{sysv3} can be solved by setting:
\[
r_1 = \sqrt{\frac{w^0-w^3}{2}}\sin{\beta_1},\quad r_2 = \sqrt{\frac{w^0+w^3}{2}}\sin{\beta_2},
\]
\[
 r_3 = \sqrt{\frac{w^0+w^3}{2}}\cos{\beta_2},\quad r_4 = \sqrt{\frac{w^0-w^3}{2}}\cos{\beta_1},
\]
and if we fix $\theta = \theta_2 - \theta_1 = \theta_4 - \theta_3$, we get from \eqref{sysv1} and \eqref{sysv2}:
\[
\sqrt{(w^0)^2-(w^3)^2} \cos\theta \cos(\beta_1+\beta_2) = w^1,\quad \sqrt{(w^0)^2-(w^3)^2} \sin\theta \cos(\beta_1+\beta_2) = w^2.
\]
The latter is solved by:
\begin{equation}\label{cstbeta}
\tan\theta = \frac{w^2}{w^1},\quad \cos(\beta_1+\beta_2) = \sqrt{\frac{(w^1)^2+(w^2)^2}{(w^0)^2-(w^3)^2}},
\end{equation}
using $\theta = \frac\pi2$ if $w^1=0$ (and we suppose here that the curve is timelike).

So the solution of this system contains several degrees of freedom as the choice of $\beta_1$,  a global phase change $\theta_i \rightarrow \theta_i + \delta$ as well as the differences $ \theta_3 - \theta_1 = \theta_4 - \theta_2 = \alpha$. In particular, if $\psi$ is a solution to the system \eqref{eqquantumflat}, then $\psi e^{i \delta}$ is so for any $\delta \in \sR$.

With $\psi$ a solution to the system \eqref{eqquantumflat}, we define the following vector in $\setC^8$:
\begin{eqnarray*}
\phi &=& (\sqrt{\chi}\;\psi,\sqrt{1-\chi}\;\psi e^{i\delta})\\
 &=& (\sqrt{\chi}\,r_1 e^{i \theta_1},\dots,\sqrt{\chi}\,r_4 e^{i \theta_4},\sqrt{1-\chi}\,r_1 e^{i (\theta_1+\delta)},\dots,\sqrt{1-\chi}\,r_4 e^{i (\theta_4+\delta)}).
\end{eqnarray*}

The condition that the matrix (also represented in a detailed form in \eqref{detailmatrix}) 
\begin{equation}\label{matrixwithflat}
-\J [D,\mathbf{a} ] =  \left( \begin{matrix}
V^{\mu} a_{,\mu} & -iV m(a-b) \\
 iV m^*(a-b) & V^{\mu} b_{,\mu}  
\end{matrix} \right) =   \left( \begin{matrix}
V^{a} a_{,a} & -iV m(a-b) \\
 iV m^*(a-b) & V^{a} b_{,a}  
\end{matrix} \right)
\end{equation}
is pointwise positive semidefinite implies the following inequality when evaluated on $\phi$:
\begin{eqnarray}
&&\chi \cdot (a\circ\gamma)^\prime  + (1-\chi) \cdot (b\circ\gamma)^\prime \nonumber\\
&=& \chi \cdot ( w^a a_{,a}) + (1-\chi)  \cdot ( w^a b_{,a})  \nonumber\\
&=&   (\sqrt{\chi}\;\psi)^* V^a \;a_{,a} (\sqrt{\chi}\;\psi)  +  (\sqrt{1-\chi}\;\psi e^{i\delta})^* V^a \;b_{,a}  (\sqrt{1-\chi}\;\psi e^{i\delta})  \nonumber\\
&\geq& \sqrt{\chi(1-\chi)} 2\Re\left\{ m(a-b) \psi^* iV \psi e^{i\delta}\right\}\nonumber\\
&=& \sqrt{\chi(1-\chi)} 2\Re\left\{ m(a-b) \sqrt{(w^0)^2-(w^3)^2} \ i e^{i\delta} \right. \nonumber\\
&&\qquad \left.  \left[ \sin\beta_1 \cos\beta_2\frac{e^{i(\theta_3 - \theta_1)}-e^{i(\theta_1 - \theta_3)}}{2i} + \cos\beta_1 \sin\beta_2\frac{e^{i(\theta_4 - \theta_2)}-e^{i(\theta_2 - \theta_4)}}{2i} \right] \right\}\nonumber\\
&=& 2\sqrt{\chi(1-\chi)} \abs{m} \abs{a-b} \sqrt{(w^0)^2-(w^3)^2} \   \sin(\beta_1+\beta_2) \sin{\alpha} \nonumber\\
\nonumber\\
&=& 2\sqrt{\chi(1-\chi)} \abs{m} \abs{a-b} \sqrt{(w^0)^2-(w^1)^2-(w^2)^2-(w^3)^2}. \label{4dineq}
\end{eqnarray}
The last but one equality is obtained using the degree of freedom $\delta$ to maximise the real part and $\alpha = \theta_3 - \theta_1 = \theta_4 - \theta_2 $. The last line follows from \eqref{cstbeta} via $\sin(\beta_1+\beta_2) = \sin\left(\arccos\sqrt{\frac{(w^1)^2+(w^2)^2}{(w^0)^2-(w^3)^2}}\right)= \sqrt{\frac{(w^0)^2-(w^1)^2-(w^2)^2-(w^3)^2}{(w^0)^2-(w^3)^2}}$ and fixing the free parameter $\sin{\alpha} = 1$.

We are now ready to prove the sufficient condition of the theorem. Let us take a~timelike curve $\gamma: [0,T] \to \M$, with $\gamma(0)=p$, $\gamma(T)=q$, $T>0$ and 
\[\sigma\,l^m_\gamma(T) = \arcsin\sqrt{\varphi}-\arcsin\sqrt{\xi},\]
with 
\begin{eqnarray*}
 l^m_\gamma(t) &=&  \int_0^t \abs{m} \sqrt{-g_{\gamma(s)}(\dot \gamma(s),\dot \gamma(s))}  ds \\
 &=& \int_0^t \abs{m} \sqrt{-g_{\mu\nu} v^\mu v^\nu}  ds =  \int_0^t \abs{m} \sqrt{-\eta_{ab} w^a w^b}  ds \\
  &=& \int_0^t \abs{m} \sqrt{(w^0)^2-(w^1)^2-(w^2)^2-(w^3)^2}  ds 
\end{eqnarray*}
and $\sigma = \sgn\left(\arcsin\sqrt{\varphi}-\arcsin\sqrt{\xi}\right) \in \set{+1,-1}$. The general inequality case is treated by transitivity of the causal order as in the proof of Proposition \ref{prop2dimsc}.

Let us take an arbitrary $\mathbf{a} = \left(\begin{smallmatrix} a & 0\\0 & b\end{smallmatrix}\right) \in \mathcal{C}$. We define the function 
$$\chi(t) = \sin^2\left( \sigma\,l^m_\gamma(t) + \arcsin \sqrt{\xi}  \right)$$
 for $t\in[0,T]$.  In the same way as in the proof of Prosition \ref{prop2dimsc}, we need to show, using the second fundamental theorem of calculus, that the following expression is always non-negative:
\begin{eqnarray*}
&&\varphi\,a(q) - \xi\, a(p) + (1-\varphi)\, b(q) - (1-\xi)\, b(p) \\
&=& \int_0^T  \chi^\prime (t)\, (a-b)(\gamma(t)) \;dt + \int_0^T \left[ \chi(t)\,(a\circ\gamma)^\prime(t)  + (1-\chi(t))\,(b\circ\gamma)^\prime(t) \right] \;dt. 
\end{eqnarray*}

Once more we have that:
\[
\chi^\prime \cdot (a-b) = \sigma \;2 \sqrt{\chi(1-\chi)} \,\left(l^m_\gamma\right)^\prime (a-b),
\]
and from the inequality \eqref{4dineq}, we conclude that:
\[
\chi \cdot (a\circ\gamma)^\prime  + (1-\chi) \cdot (b\circ\gamma)^\prime \geq 2 \sqrt{\chi(1-\chi)}  \left(l^m_\gamma\right)^\prime \abs{a-b},
\]
which completes the sufficient condition.

The proof of the necessary condition is an exact generalisation of the proof of Proposition \ref{prop2dimnc} to a curve $\gamma$ on the 4-dimensional space $\M$ and we will not repeat the complete proof here. The only difference is in the checking that the particular element $\mathbf{a} = \left(\begin{smallmatrix}a & 0\\0 & b\end{smallmatrix}\right)$ defined along the curve $\gamma$ (and considering its implicit extension to whole $\M$) by 
\begin{eqnarray}
a(\gamma(t)) &=&  - \frac 12 \cot( l^m_\gamma(t) + \sigma \arcsin\sqrt{\xi} + \epsilon), \nonumber \\
b(\gamma(t)) &=&   \frac 12 \tan( l^m_\gamma(t) + \sigma \arcsin\sqrt{\xi} + \epsilon), \label{defab4d}
\end{eqnarray}  
respects the conditions of a causal element. By the argument outlined on page \pageref{matrixEN} it is sufficient to check the positive semidefiniteness of the matrix \eqref{matrixwithflat} on every point $p$ of the curve $\gamma$. To simplify the notation and avoid a confusion between the curved ($^\mu$) and flat ($^a$) indices, we will work with a local trivialisation on a open set $U_p$ using the local flat coordinates $ds^2 \vert_{p} = \eta_{ab} dx^a dx^b$. This technical computation is presented in the Appendix B.
\end{proof}

\vspace*{0.3cm}

\section{The impact of fluctuations}\label{secscalar}

The models presented in Sections \ref{sec2dim} and \ref{sec4dim} use a standard definition of the Dirac operator for an almost commutative space, with a minimal coupling between the finite part and the manifold. However, the axioms of noncommutative geometry allow us to construct more general Dirac operators by taking into account the inner fluctuations \cite[Section 10.8]{MC08}.

Since the spectral triple at hand does not admit a reality structure, the most general Dirac operator for a two-sheeted space-time reads
\begin{align*}
D_A = D + A, && A = \sum_{\text{finite}} \mathbf{a}_i [D,\mathbf{b}_i],
\end{align*}
with $\mathbf{a}_i, \mathbf{b}_i \in \A$ and such that $i A$ is Krein-self-adjoint. As usually in the almost commutative case, the fluctuation term $A$ splits into two parts (see for instance \cite[Section 2.5.1]{Dungen}) and we have
\begin{align}\label{D_fluctuated}
D_A = D_{\M} \ox 1 + \widetilde{\gamma}^{\mu} \ox \left(\begin{smallmatrix} A_{\mu} & 0 \\ 0 & B_{\mu} \end{smallmatrix} \right) + \gamma_{\M} \ox \left( \begin{smallmatrix} 0 & \Phi_{ } \\ \Phi^*_{ } & 0 \end{smallmatrix} \right),
\end{align}
with $A_{\mu}, B_{\mu}, \Phi \in \widetilde\A_\M$. We shall call the second and the third term of \eqref{D_fluctuated} the vector and scalar fluctuations respectively.

We now investigate the impact of the fluctuations on the causal structure of a two-sheeted space-time. We adopt the name of a \emph{fluctuated two-sheeted space-time} for the Lorentzian spectral triple from Definition \ref{def:twosheeted} with $D$ exchanged for $D_A$. The results, which hold both in dimension 2 and 4, are summarised in the following theorem.

\begin{thm}\label{thm:fluct}
Let $\preceq$ denote the partial order determined by the causal cone associated with a fluctuated Dirac operator on a two-sheeted space-time. Two states $\omega_{p,\xi}, \omega_{q,\varphi} \in \N(\A)$ are causally related with $\omega_{p,\xi} \preceq \omega_{q,\varphi}$ if and only if  $p \preceq q$ on $\M$ and
\begin{align}\label{const:scalar}
\int_0^t \abs{\Phi(\gamma(s))} \sqrt{-g_{\gamma(s)}(\dot \gamma(s),\dot \gamma(s))}  \;ds   \geq \abs{\arcsin\sqrt{\varphi}-\arcsin\sqrt{\xi}}.
\end{align}
\end{thm}

\vspace*{0.3cm}

\begin{proof}
Let us first note that vector fluctuations have no impact on the causal cone. Indeed, since our algebra $\widetilde\A$ is commutative, the term $A^V \vc \widetilde\gamma^{\mu} \ox \left( \begin{smallmatrix} A_{\mu} & 0 \\ 0 & B_{\mu} \end{smallmatrix} \right) $ will commutate with every element $\mathbf{a} \in \widetilde\A$ and thus
\begin{align*}
\J [D+A^V,\mathbf{a}] = \J [D,\mathbf{a}].
\end{align*}

On the other hand, the scalar fluctuation will affect the causal structure. The formula \eqref{const:scalar} is obtained simply by replacing the function $l^m_\gamma$ (see \eqref{lm}) by 
\[l^\Phi_\gamma(t) \vc  \int_0^t \abs{\Phi(\gamma(s))} \sqrt{-g_{\gamma(s)}(\dot \gamma(s),\dot \gamma(s))} \,  ds \]
in Proposition \ref{prop2dimsc}, Proposition \ref{prop2dimnc} and Theorem \ref{thm4d}.
\end{proof}

\vspace*{0.3cm}

It is an interesting phenomenon that the scalar fluctuation, which yields the Higgs field in the full noncommutative Standard Model, affects the causal structure in the same way as a conformal transformation of the space-time metric with $\Omega = \abs{\Phi}^{-1}$. The interplay between the Higgs field and conformal transformations was discussed in \cite{Dungen}, though at the level of spectral action.

Let us stress however that the field $\Phi$ is external from the viewpoint of the space-time metric.  It means that the causal relation becomes dependent on the choice of the path between the two space-time points $p$ and $q$ since maximal geodesics do not give automatically the best constraint, as it was the case for a constant scalar field.

\vspace*{0.3cm}

\section{Conclusions}\label{sec:conc}

In this paper, we studied the causal structure of a particular Lorentzian almost commutative model called \emph{two-sheeted space-time}. The computation, summarised in Theorem \ref{thm_MAIN}, was done explicitly in the case of dimensions 2 and 4. Surprisingly enough, it turned out that causal relations between the two disjoint sheets are possible, under the condition that the amount of proper time (i.e.~the length of a causal curve) is sufficiently large. We also showed that fluctuations of the Dirac operator result in a conformal-like weighting of the proper time necessary for a causal evolution from one sheet to another.

We note that the procedure highlighted at the beginning of Section \ref{sec4dim} can be applied to our previous model presented in \cite{SIGMA2014} and based on the algebra $M_2(\setC)$. This would allow to extend the results, proved in \cite{SIGMA2014} only for a 2-dimensional Minkowski space-time, to general even dimensional globally hyperbolic space-times. Moreover, the same arguments as in the proof of Theorem \ref{thm:fluct} could be applied to include scalar fluctuations in the Dirac operator. On the other hand, since the algebra $\A_F = M_2(\setC)$ is not commutative, the vector fluctuations will, in general, affect the causal structure of this almost commutative space-time.

Our results show that the peculiar causal properties of almost commutative space-times are universal, not being an artefact of 2-dimensionality or flatness of the underlying space-time. However, the complete computation for an arbitrary $n$-dimensional model is difficult for the following two reasons. Firstly, the system of equations \eqref{eqquantum} (which can be seen as a quantum problem) becomes cumbersome for a large number of degrees of freedom, despite being highly underdetermined. Secondly, checking if the matrix \eqref{matrixEN} is positive semidefinite for an arbitrary dimension is directly related to an old problem --- the \emph{Matrix Vieta Theorem} \cite{fuchs1995matrix,connes1997matrix} --- which still remains unsolved, except in a perturbative way \cite{schwarz2000noncommutative}. A solution to these two mathematical problems, which are in fact independent from our work, would open the possibility to compute causal structures of several different models of arbitrary dimensions.

\vspace*{0.3cm}

\section*{Acknowledgements}

NF was partially supported within a grant from the John Templeton Foundation. ME and NF were supported within the project NCN PRELUDIUM 2013/09/N/ST1/01108. ME acknowledges the support of the Marian Smoluchowski Krak\'ow Research Consortium ``Matter--Energy--Future'' within the programme KNOW.

The authors would like to thank the Hausdorff Institute in Bonn for hospitality during the Trimester Program ``Non-commutative Geometry and its Applications'', in the course of which a part of this research was carried out.

\vspace*{0.3cm}

\appendix
\section*{Appendix A}
\setcounter{section}{1}

In this Appendix, we show that the element $\mathbf{a}$ defined in the proof of Proposition \ref{prop2dimnc} (with $a$ and $b$ defined in \eqref{defab}) is a causal element along the curve $\gamma$, i.e.~that the matrix given in Lemma \ref{propC}$(b)$ is positive semidefinite.

The partial derivatives of $a$ are defined in such a way that its directional derivative is maximal along $\gamma$ and corresponds to
\[
\dv{}{t}(a\circ\gamma)= -\frac 12 \dv{}{t} \cot \theta = \abs{m}\sqrt{\lambda_1\lambda_2} \csc^2 \theta,
\]
where $\theta=l^m_\gamma(t) + \sigma \arcsin\sqrt{\xi} + \epsilon$ and using the functions \eqref{defl1l2} with the relation \eqref{rell1l2}. Then, we have the following values for the partial derivatives:
\begin{itemize}
\item $\Omega (a_{,0}+a_{,1} ) =\frac12\sqrt{\frac{\lambda_2}{\lambda_1}}\abs{m} \csc^2 \theta$,
\item $\Omega (a_{,0}-a_{,1} ) =\frac12\sqrt{\frac{\lambda_1}{\lambda_2}}\abs{m} \csc^2 \theta$,
\end{itemize}
since $\dv{}{t}(a\circ\gamma) =  \lambda_1 \Omega (a_{,0}+a_{,1}) + \lambda_2 \Omega (a_{,0}-a_{,1}).$

A similar reasoning with $b$ gives 
\[
\dv{}{t}(b\circ\gamma)= \frac 12 \dv{}{t} \tan \theta = \abs{m}\sqrt{\lambda_1\lambda_2} \sec^2 \theta
\]
and
\begin{itemize}
\item $\Omega(b_{,0}+b_{,1}) =\frac12\sqrt{\frac{\lambda_2}{\lambda_1}}\abs{m} \sec^2 \theta$,
\item $\Omega(b_{,0}-b_{,1}) =\frac12\sqrt{\frac{\lambda_1}{\lambda_2}}\abs{m} \sec^2 \theta$.\\
\end{itemize}
Since
\[
a-b = - \frac{1}{2 \sin\theta\cos\theta }=  -\csc 2\theta,
\]
\vspace{0.2cm}

we must prove that the following matrix is positive semidefinite:
\vspace{0.2cm}
\[
A=\left(
\begin{matrix}
\frac12\sqrt{\frac{\lambda_2}{\lambda_1}}\abs{m} \csc^2 \theta & 0 & 0 & m \csc 2\theta \\
0 & \frac12\sqrt{\frac{\lambda_1}{\lambda_2}}\abs{m} \csc^2 \theta & -m \csc 2\theta & 0 \\
0 & - m^* \csc 2\theta & \frac12\sqrt{\frac{\lambda_2}{\lambda_1}}\abs{m} \sec^2 \theta & 0 \\
m^* \csc 2\theta & 0 & 0 & \frac12\sqrt{\frac{\lambda_1}{\lambda_2}}\abs{m} \sec^2 \theta
\end{matrix}
\right).\\
\]
\vspace{0.2cm}

Since the eigenvalues of $A$ are the roots of the characteristic polynomial:
\[
\det(A-\lambda1)=\lambda^4-c_1\lambda^3+c_2\lambda^2-c_3\lambda+c_4,
\]
from Vieta's formulas it is sufficient to check that $c_k \geq 0$ for $k=1,\dots,4$.\\

An explicit computation (using e.g.~Mathematica) gives the following:
\begin{eqnarray*}
c_1&=&\tr A=2\abs{m} \left(\sqrt{\frac{\lambda_2}{\lambda_1}} + \sqrt{\frac{\lambda_1}{\lambda_2}} \right) \csc^2 2\theta  \geq 0,
\\
c_2&=&\frac12\prt{\prt{\tr A}^2-\tr A^2}=\frac{\abs{m}^2 \left[ (\lambda_1-\lambda_2)^2 + 4 \lambda_1\lambda_2 \csc^2 2\theta\right] \csc^2 2\theta}{\lambda_1\lambda_2} \geq 0,
\\
c_3&=&\frac16\prt{\prt{\tr A}^3-3\tr A^2\tr A+2\tr A^3}=0,
\\
c_4&=&\det(A)=0.
\end{eqnarray*}
\vspace{0.2cm}

Hence the matrix $a$ is positive semidefinite and respects the conditions of Lemma \ref{propC}.

\vspace*{0.5cm}

\section*{Appendix B}
\setcounter{section}{2}

In this Appendix, we regroup the technical definitions and computations of the proof of Theorem \ref{thm4d} in Section \ref{sec4dim}.\\

The representation of the flat gamma matrices is chosen to be the following one (Weyl representation):

\[
\gamma^0=
\left( \begin{array}{cccc}
0 & 0 & i & 0 \\ 
0 & 0 & 0 & i \\ 
i & 0 & 0 & 0 \\ 
0 & i & 0 & 0 
\end{array} \right),
\qquad
\gamma^1=
\left( \begin{array}{cccc}
0 & 0 & 0 & i \\ 
0 & 0 & i & 0 \\ 
0 & -i & 0 & 0 \\ 
-i & 0 & 0 & 0 
\end{array} \right),
\]
\[
\gamma^2=
\left( \begin{array}{cccc}
0 & 0 & 0 & 1 \\ 
0 & 0 & -1 & 0 \\ 
0 & -1 & 0 & 0 \\ 
1 & 0 & 0 & 0 
\end{array} \right),
\qquad
\gamma^3=
\left( \begin{array}{cccc}
0 & 0 & i & 0 \\ 
0 & 0 & 0 & -i \\ 
-i & 0 & 0 & 0 \\ 
0 & i & 0 & 0 
\end{array} \right). \]

\vspace{0.3cm}

In this representation, the (flat) operators present in system \eqref{eqquantumflat} are explicitly given by:
\[
V^0 = -\gamma^0\gamma^0=
\left( \begin{array}{cccc}
1 & 0 & 0 & 0 \\ 
0 & 1 & 0 & 0 \\ 
0 & 0 & 1 & 0 \\ 
0 & 0 & 0 & 1 
\end{array} \right),
\qquad
V^1 = -\gamma^0\gamma^1=
\left( \begin{array}{cccc}
0 & -1 & 0 & 0 \\ 
-1 & 0 & 0 & 0 \\ 
0 & 0 & 0 & 1 \\ 
0 & 0 & 1 & 0 
\end{array} \right),
\]
\[
V^2 = -\gamma^0\gamma^2=
\left( \begin{array}{cccc}
0 & i & 0 & 0 \\ 
-i & 0 & 0 & 0 \\ 
0 & 0 & 0 & -i \\ 
0 & 0 & i & 0 
\end{array} \right),
\qquad
V^3 = -\gamma^0\gamma^3=
\left( \begin{array}{cccc}
-1 & 0 & 0 & 0 \\ 
0 & 1 & 0 & 0 \\ 
0 & 0 & 1 & 0 \\ 
0 & 0 & 0 & -1 
\end{array} \right), 
\]
\[
iV = -\gamma^1\gamma^2\gamma^3=
\left( \begin{array}{cccc}
0 & 0 & 1 & 0 \\ 
0 & 0 & 0 & 1 \\ 
-1 & 0 & 0 & 0 \\ 
0 & -1 & 0 & 0 
\end{array} \right).
\]

\vspace{0.3cm}

The matrix \eqref{matrixwithflat} is given explicitly in locally flat coordinates around a point $p$ (i.e.~$ds^2_{|p} = \eta_{ab} dx^a dx^b$) by: 
\begin{equation}\label{detailmatrix}
A = -\J[D,\mathbf{a}] =
\end{equation}
\[
 \left(
{
\begin{smallmatrix}
 a_{,0} - a_{,3} & -a_{,1} + i a_{,2}  & 0 &  0 &  0 &  0 &  -m(a-b) &  0
\\
-a_{,1} - i a_{,2}  &  a_{,0} + a_{,3} & 0 &  0 &  0 &  0 &  0 &  -m(a-b) 
\\
0 &  0 &  a_{,0} + a_{,3} & a_{,1} - i a_{,2}  &  m(a-b) &  0 &   0&  0 
\\
0 &  0 &  a_{,1} + i a_{,2} & a_{,0} - a_{,3}  &  0  &  m(a-b) &  0 &   0
\\
 0 &  0 &  m^*(a-b) &  0 & b_{,0} - b_{,3} & -b_{,1} + i b_{,2}  & 0 &  0 
\\
0 &  0 &  0 &  m^*(a-b)  & -b_{,1} - i b_{,2}  &  b_{,0} + b_{,3} & 0 &  0 
\\
 -m^*(a-b) &  0 &   0 &  0  & 0 &  0 &  b_{,0} + b_{,3} & b_{,1} - i b_{,2} 
\\
0  &  -m^*(a-b) &  0 &   0 & 0 &  0 &  b_{,1} + i b_{,2} & b_{,0} - b_{,3}  
\\
\end{smallmatrix}
}
\right)\\
\]

\vspace{0.5cm}

For the proof of the sufficient condition of Theorem \ref{thm4d}, we need to check that the coefficients of the characteristic polynomial of this matrix are everywhere non-negative with the functions $a$ and $b$ replaced by the ones defined in \eqref{defab4d}, which gives the following elements as entries of the matrix:
\begin{align*}
a_{,0} \pm a_{,3} &= (w^0 \pm w^3) \frac{  \abs{m}  \csc^2 \theta}{2 \sqrt{(w^0)^2 - (w^1)^2 - (w^2)^2 -(w^3)^2}},\\
 a_{,1} \pm i a_{,2} &= (w^1 \pm i w^2) \frac{  \abs{m}  \csc^2 \theta}{2 \sqrt{(w^0)^2 - (w^1)^2 - (w^2)^2 -(w^3)^2}}, \\
b_{,0} \pm b_{,3} &= (w^0 \pm w^3) \frac{  \abs{m}  \sec^2 \theta}{2 \sqrt{(w^0)^2 - (w^1)^2 - (w^2)^2 -(w^3)^2}}, \\
 b_{,1} \pm i b_{,2} &= (w^1 \pm i w^2) \frac{  \abs{m}  \sec^2 \theta}{2 \sqrt{(w^0)^2 - (w^1)^2 - (w^2)^2 -(w^3)^2}}, \\
m(a-b) &= -m\csc 2\theta,\\
 m^*(a-b) &= -m^*\csc 2\theta.
\end{align*}

An explicit computation (using e.g.~Mathematica) gives the following:
\begin{eqnarray*}
c_1&=&\mathsmaller{\frac{ 8  \abs{m}(w^0)\, \csc^2 2\theta}{\sqrt{(w^0)^2 - (w^1)^2 - (w^2)^2 -(w^3)^2}}}  \geq 0,
\\
c_2&=&\mathsmaller{\frac{  4 \abs{m}^2 (6 (w^0)^2  - (w^1)^2 - (w^2)^2 - (w^3)^2 - ((w^1)^2 + (w^2)^2 + (w^3)^2) \cos 4\theta )  \csc^4 2\theta}{(w^0)^2 - (w^1)^2 - (w^2)^2 -(w^3)^2} }  \geq 0,
\\
c_3&=&\mathsmaller{\frac{  16 \abs{m}^3 (2 (w^0)^2  - (w^1)^2 - (w^2)^2 - (w^3)^2 - ((w^1)^2 + (w^2)^2 + (w^3)^2) \cos 4\theta )  \csc^6 2\theta}{((w^0)^2 - (w^1)^2 - (w^2)^2 -(w^3)^2)^{\frac 32}}}  \geq 0,
\\
c_4&=&\mathsmaller{\frac{   \abs{m}^4 (2 (w^0)^2  - (w^1)^2 - (w^2)^2 - (w^3)^2 - ((w^1)^2 + (w^2)^2 + (w^3)^2) \cos 4\theta )^2  \csc^8 2\theta}{((w^0)^2 - (w^1)^2 - (w^2)^2 -(w^3)^2)^{\frac 32}}}  \geq 0,
\\
c_5&=& 0,\qquad c_6\ =\ 0,\qquad c_7\ = \ 0,\qquad c_8\ = \ 0,
\end{eqnarray*}
since $(w^0)^2 - (w^1)^2 - (w^2)^2 -(w^3)^2 = -g_{\mu\nu} v^{\mu}v^{\nu} \geq 0$.

\vspace{0.5cm}

\bibliographystyle{elsarticle-num}
\bibliography{causality_bib}{}

\end{document}